\documentclass{llncs}
\usepackage{makeidx}  % allows for indexgeneration
\usepackage[leqno]{amsmath}
\usepackage{amssymb}
\usepackage{amsfonts}
\usepackage{amsmath}
\usepackage{stmaryrd}
\usepackage{epsfig}
\usepackage{color}
\usepackage{colordvi}
\usepackage{ifthen} % para TODOs
\usepackage{proof}

\newcommand{\todoFlag}{ON} % si vale ON aparecen los TODO's en cajas, sino desaparecen
\newcommand{\dudaFlag}{ON} %

\newcommand{\nc}{\newcommand}
\nc{\obs}{\mathcal{O}} % use $\obs$
\nc{\prog}{\mathcal{P}}
\nc{\den}[1]{[\![#1]\!]}  % use $\den{e}$

\newcommand{\denn}[1]{[\![\![#1]\!]\!]}

\nc{\ordSus}{\lesssim} % orden de subsuncion de sustituciones
\nc{\ordSuse}{\precnsim} % orden de subsuncion de sustituciones estricto
\nc{\ordPos}{\lesssim} % orden de subsuncion de posiciones
\nc{\ordap}{\sqsubseteq} % orden de aproximacion de expresiones parciales de crwl
\nc{\rw}{\to} % traditional rewriting
\nc{\tor}{\to} % flecha de (l -> r) de las reglas de programa
\nc{\crwlto}{\rightarrowtriangle}  % CRWL-approximation e -> t
\nc{\clto}{\crwlto}                % CRWL-approximation e -> t
\nc{\conscrwl}{\vdash_{\it CRWL}}                % CRWL-derivability
\nc{\clp}{{\cal P} \vdash_{CRWL^{+}}} % variacion de CRWL,  CRWL^+
\nc{\cldt}{{\cal P} \vdash_{CRWL^{d}}} % variacion de CRWL, CRWL^d
\nc{\dend}[1]{\den{#1}^d} % denotacion de CRWL^d
\nc{\dg}[1]{\den{#1}_{CRWL}}
\nc{\cl}{{\cal P} \vdash_{CRWL}}
\nc{\clho}{{\cal P} \vdash_{HOCRWL}}
\nc{\gl}{{\cal P} \vdash_{CRWL_{let}}}
\nc{\dgl}[1]{\den{#1}_{CRWL_{let}}}
\nc{\dcl}[1]{\dgl{#1}}
\nc{\ddgl}[1]{\denn{#1}_{CRWL_{let}}}
\nc{\ddcl}[1]{\ddgl{#1}}
\nc{\f}{{\to^{l}}}    %let-rewriting
\nc{\fnf}{\to_{lnf}}    %let-rewriting sin Fapp
\nc{\fd}{\to_{lw}}    %let-rewriting con swap
\nc{\fnt}{{\to^{L}}} % let rewriting no terminante con (LetIn')
\nc{\nr}{\leadsto} % narrowing clasico
\nc{\fnr}{{\leadsto^L}}  % let-rewriting  con narrowing
\nc{\fnre}{\leadsto^{l^*}}  % let-rewriting  con narrowing, varios pasos, lo he cambiado a peque
\nc{\fnrc}[1]{\leadsto^{L^{#1}}}  % let-rewriting  con narrowing, varios pasos
\nc{\fnrl}{{\leadsto^l}}  % let-rewriting  con narrowing, relacion mas pequenia

\nc{\var}{{\cal V}}
\nc{\ra}{\tor}
\nc{\leqhyp}{\Subset}
\nc{\tot}[1]{{#1}^\tau}
\nc{\tlr}[1]{\widehat{#1}} % transformacion de LExp a Exp
\nc{\jn}{\Join} % lazo de la join
\nc{\tr}{\underline{\mbox{\textbf{t}}}}
\nc{\slt}{\hookleftarrow} % sharing let
\nc{\clt}{\hookrightarrow} % copy let
\nc{\con}{{\cal C}}  %  contexto
\nc{\cnn}[1]{\con[#1]}  % argumento en contexto
\nc{\cnnp}[1]{\con'[#1]}  % argumento en contexto'
\nc{\nat}{{\mathbb N}}
\nc{\progh}{\hat{\prog}} % un programa cualquiera transformado
\nc{\progm}{{\cal M}} % un programa
\nc{\progr}{{\cal R}} % un programa
\nc{\ps}{\vdash}    % prueba en crwl
\nc{\pps}{\vdash_{\pi CRWL}}    % prueba en plural semantics
\nc{\pcrwl}{{\it $\pi$CRWL}} % plural semantics
\nc{\crwl}{{\it CRWL}}
\nc{\denp}[1]{\den{#1}^{pl}} % denotacion en plural semantics
\nc{\dens}[1]{\den{#1}^{sg}} % denotacion en CRWL, explicitando que es singular
\nc{\denr}[1]{\den{#1}^{rw}} % denotacion en reescritura
\nc{\pst}[1]{pST(#1)} % plural semantics transformation
\nc{\icsus}{CSubst_\perp^?}% csustituciones con interrogacion
\nc{\icterm}{CTerm^?} % cterminos con interrogacion
\nc{\ppop}{\rightarrowtail} % flecha popeye de plural
\nc{\pordap}{\ordap_{\pi}} % orden de aproximacion estilo multicjtos para \icsus
\nc{\dsord}{\unlhd} % orden entre las denotaciones de las sustituciones
\nc{\cltop}{\clto} % flecha de plural semantics
\nc{\dsordp}{\dsord_{M}} % orden entre las denotaciones de las sustituciones con cjtos de valores
\nc{\pl}{\pi } % simbolo de pluralidad, que uso por ej para las s-sustituciones
\nc{\cjto}{{\cal A}} % simbolo para referirse a un cjto cualquiera
\nc{\cjtoD}{{\cal S}} % simbolo para referirse a un cjto cualquiera
\nc{\vran}{vran}
\nc{\pos}[1]{{\cal O}(#1)} % conjuntos de posiciones en un termino
\nc{\posV}[1]{\tilde{{\cal O}}(#1)} % conjuntos de posiciones de variables en un termino

\nc{\ds}{{\cal D}} % Conjunto de las denotaciones = partes de CTerm_{\perp}
\nc{\hds}{{\cal HD}} % Conjunto de las hyperdenotaciones
\nc{\ohs}{\leqhyp} % orden de las hipersemantica
\nc{\hde}{\varphi} % elemento del conjunto de las hyperdenotaciones
\nc{\de}{\delta} % elemento del conjunto de las denotaciones
\nc{\uhs}{{\ \Cup\ }} % hiper union de hipersemanticas
\nc{\Uhs}{{\Cup^\Cup}} % hiper union grande de hipersemanticas
\nc{\dennD}[1]{{\triangleleft\!\! \parallel \!{#1}\!\parallel \!\!\triangleright}} % descomposicion de la hipersemantica
\nc{\denD}[1]{{\triangleleft \!|{#1}| \!\triangleright}} % descomposicion de la semantica
\nc{\hspm}{hspm~} % funcion para construir hipersemanticas con pattern matching

\nc{\obsig}{{\obs_{\Theta}}} % Observable de variables uno
\nc{\obsigp}[1]{{\obs^{{#1}}_{\Theta}}} % Observable de variables uno, bajo un programa
\nc{\obsigfo}{{\obs_{{fo}_{\Theta}}}} % Observable de variables FO
\nc{\obsigfop}[1]{{\obs^{{#1}}_{{fo}_{\Theta}}}} % Observable de variables FO, bajo un programa
\nc{\obsfo}{{\obs_{fo}}}
\nc{\obsfop}[1]{{\obs^{{#1}}_{fo}}}
\nc{\ppp}{{\cal P}}
\nc{\fs}{{FS}}

\nc{\todoLong}[1]{\ifthenelse{\equal{\todoFlag}{ON}}{~\\{\fcolorbox{red}{yellow}{\begin{minipage}{.985\textwidth}{#1}\end{minipage}}}}{}
} % para TODO's, depende de \todoFlag
\nc{\todo}[1]{\ifthenelse{\equal{\todoFlag}{ON}}{\textcolor{red} {#1}}{}} % para TODO's, depende de \todoFlag
\nc{\todobis}[2]{\ifthenelse{\equal{\todoFlag}{ON}}{~\\{\fcolorbox{red}{#1}{\begin{minipage}{.985\textwidth}{#2}\end{minipage}}}}{}
} % \todobis{color}{contenido}para TODO's, depende de \todoFlag
\nc{\duda}[1]{\ifthenelse{\equal{\dudaFlag}{ON}}{~\\{\fcolorbox{yellow}{Azure}{\begin{minipage}{.985\textwidth}{#1}\end{minipage}}}}{}
}% cosas dudosas

\begin{document}
%%%%%%%%%%%%%%%%
% \input{frontmatter}
\frontmatter
\pagestyle{headings}  % switches on printing of running heads
% \addtocmark{Declarative Debugging Datalog Programs} % additional mark in the TOC
%
\mainmatter              % start of the contributions
\title{The Full Abstraction Problem for Higher Order  Functional-Logic Programs}
\titlerunning{The Full Abstraction Problem ...}  % abbreviated title (for running head)
%                                     also used for the TOC unless
%                                     \toctitle is used
%
% \author{R. Caballero\inst{1} \and F.J. L\'opez-Fraguas\inst{1} and J. Rodr\'iguez-Hortal\'a\inst{1}
\author{F.J. L\'opez-Fraguas and J. Rodr\'iguez-Hortal\'a
        \thanks{This work has been partially supported by the Spanish
 projects TIN2005-09207-C03-03, TIN2008-06622-C03-01, S-0505/TIC/0407
 and UCM-BSCH-GR58/08-910502.}}
%
% \authorrunning{R. Caballero et al.}   % abbreviated author list (for running head)
%
%%%% list of authors for the TOC (use if author list has to be modified)
%\tocauthor{R. Caballero, F.J. L\'opez-Fraguas }
%
\institute{Departamento de Sistemas Inform{\'a}ticos y
Computaci{\'o}n \\
%Universidad Complutense de Madrid, Spain\\
% \email{rafa@sip.ucm.es}, \email{fraguas@fdi.ucm.es}, \email{juanrh@fdi.ucm.es}}
\email{fraguas@fdi.ucm.es}, \email{juanrh@fdi.ucm.es}}

\maketitle              % typeset the title of the contribution

%%%%%%%%%%%%%%%%%%%%%%%%%%%%%%%%
% \input{abstract}
\begin{abstract}
Developing suitable formal semantics can be of great  help in the
understanding, design and implementation of a programming language,
and act as a guide for software development tools like analyzers or
partial evaluators. In this sense, full abstraction is a highly
desirable property, indicating a perfect correspondence between the
semantics and  the observable behavior of program pieces. In this
work we address the question of full abstraction for the family of
modern functional logic languages, in which functions can be  higher
order and non-deterministic, and where the semantics adopted for
non-determinism is \emph{call-time choice}. We show that, with
respect to natural notions of \emph{observation}, any semantics
based on \emph{extensional} functions is necessarily unsound; in
contrast, we show that the higher order version of \emph{CRWL}, a
well-known existing semantic framework for functional logic
programming, based on  an \emph{intensional} view of functions,
turns out to be fully abstract and compositional.
% Moreover, we extend the results to the case of observation notions that give an active role to the variables in an expression, by considering the substitutions that could be applied to them.
\end{abstract}
%%%%%%%%%%%%%%%%%%%%%%%%%%%%%%%%
%%%%%%%%%%%%%%%%%%%%%%%%%%%%%%%%
%%%%%%%%%%%%%%%%%%%%%%%%%%%%%%%%%%%%%%%%%%%%%%%%5
\section{Introduction}\label{sec:intro}
Developing suitable formal semantics can  be of great help in the
understanding, design and implementation of a programming language,
and acts as a guide for software development tools like analyzers or
partial evaluators. In this sense, \emph{full abstraction} is a
highly desirable property, indicating a perfect correspondence
between the semantics and the behavior of program pieces, according
to a given criterion of \emph{observation}.

The notion of full abstraction was introduced by Plotkin
\cite{Plotkin77} in connection to PCF, a simple model of functional
programming based on $\lambda$-calculus.  He realized that the
standard Scott semantics, in which expressions of functional types
have classical mathematical functions as meanings, lacks full
abstraction with respect to observing the value obtained in the
evaluation of an expression. The reason lays in the impossibility of
defining the function \emph{por} (\emph{parallel or}) in PCF.
%, an inherently sequential language, as $\lambda$-calculus is.
Using this
fact one can build two higher order (HO) expressions $e_1,e_2$
denoting two different mathematical  functions
$\varphi_1,\varphi_2$,
both expecting boolean functions as arguments,
such that $\varphi_1,\varphi_2$ only differ when applied
to \emph{por} as argument. Therefore $e_1,e_2$ have different Scott
semantics but this difference cannot be \emph{observed}. It is
usually said that the semantics is \emph{ too concrete}. Notice,
however, that Scott semantics for PCF is \emph{sound}, that is, if
two expressions have the same semantics, they cannot be observably
distinguished. Unsoundness of a semantics can be considered a flaw,
much more severe that being too concrete, which is more a
weakness than a flaw.

Full abstraction for PCF was achieved in different technical ways
(see e.g. \cite{BerryCL86}). But for our purposes it is more interesting to
recall that the Scott semantics becomes fully abstract if PCF is
enriched with the `missing' \emph{ por} function (see e.g. \cite{Mitchell96}).
The mainstream of functional logic programming (FLP, see \cite{Hanus07ICLP}) is
based rather in the theory of term rewriting systems than in
$\lambda$-calculus; a consequence is that parallel or can be defined
straightforwardly  by an overlapping (almost orthogonal) rewriting
system. So one could think of assigning to FLP languages a
denotational semantics in the FP style. For instance, for a
definition like $f~x=0$, one could assign to $f$ the meaning
$\lambda x.0$.
The next step of our discussion is taking into account that modern
FLP languages like Curry \cite{Han06curry} or Toy \cite{LS99}
also permit non-confluent and non-terminating programs that define
non-deterministic non-strict functions. This suggests that the standard
semantics should be modified in the sense that the meaning of a
function would be some kind of set-valued function.
%It might seem
%that  an appropriate choice of \emph{powerdomain} could give the
%adequate technical support to this point of view.

The starting motivation of this paper is that \emph{this roadmap
cannot be followed anymore when non-determinism is combined with
HO}, at least when considering \emph{call-time choice} \cite{hussmann93,GHLR99}, which is the
notion of non-determinism adopted in, e.g., Curry or Toy. The following
example taken from \cite{LRSflops08} shows it:

\begin{example}\label{fdouble}
The following program computes with natural numbers represented by
the constructors $0$ and $s$, and where $+$ is defined as usual. The
syntax uses HO curried notation.
\begin{verbatim}
     g X -> 0         f -> g          f' X -> f X
     h X -> s 0       f -> h

     fadd F G X -> (F X) + (G X)      fdouble F -> fadd F F
\end{verbatim}
Here $f$ and $f'$ are non-deterministic functions that are (by
definition of $f'$) extensionally equivalent. In a set-valued variant of Scott semantics,
their common denotation would be the function $\lambda
X.\{0,s ~0\}$, or something essentially equivalent. But this leads to
unsoundness of the semantics. To see why, consider  the expressions
\emph{(fdouble f 0)} and \emph{(fdouble f' 0)}. In Curry  or Toy,
the possible values for \emph{(fdouble f 0)} are\emph{ 0, s (s 0)},
while \emph{(fdouble f' 0)} can be in addition reduced to  \emph{s
0}. The operational reason to this situation is that \emph{fdouble f 0} is rewritten first to
\emph{fadd f f 0} and then to \emph{f 0 + f 0};  now,  call-time
choice enforces that  evaluation of the two created
copies of $f$ (which is an evaluable expression) must be shared.  In
the case of \emph{f' 0 + f' 0}, since $f'$ is a normal form, the two
occurrences of \emph{f' 0} evolve independently. We see then that
$f$ and $f'$ can be put in a context able to distinguish them, implying
that any semantics assigning $f$ and $f'$ the same denotation is
necessarily unsound, and therefore not fully abstract.
\end{example}

The combination \emph{HO + Non-determinism + call-time choice} was
addressed in \emph{HOCRWL} \cite{GHR97,GHR99}, an extension to HO of
\emph{CRWL}  \cite{GHLR99}, a semantic FO framework specifically
devised for FLP with call-time choice semantics for non-determinism.
\emph{HOCRWL} adopts an \emph{intensional} view of functions, where
different descriptions --in the form of \emph{HO-patterns}-- of the
same extensional function are distinguished as different data.
The intensional point of view of {\it HOCRWL} was an \emph{a priori} design decision,
motivated by the desire of
achieving enough power for HO programming  while avoiding the
complexity of higher-order unification of $\lambda$-terms modulo
$\beta\eta$, followed in other approaches \cite{Miller91,HanPre99}. The issues of
soundness or full abstraction were not the (explicit nor implicit) concerns of \cite{GHR97,GHR99};
whether {\it HOCRWL} actually fulfils those properties or not is
exactly the question considered in this paper.
As we will get positive answers, an anticipated conclusion of our
work is that one must take into account
intensional descriptions of functions
%intentions
as  sensible meanings of expressions in HO non-deterministic FLP programs,
even if one does not want to explicitly program with HO-patterns% (a feature offered by HOCRWL, but unpleasant to some people)
.

The rest of the paper is organized as follows.
Next section recalls some essential preliminaries about applicative HO rewrite systems
and the {\it HOCRWL} framework. We introduce also some terminology about semantics
and extensionality needed for Sect. \ref{sec:fully}, where we examine soundness and full abstraction
with respect to reasonable  notions of observation based on the result of reductions.
The section ends with a discussion of the problems encountered when programs have {\em extra} variables, i.e., variables occuring in right, but not in
left-hand sides of function defining rules.
%We start Sect. \ref{sec:variables} by discussing the limitations of the results obtained so far,
%in relation to the role of variables in expressions an programs, and improve our results
%in this direction with new powerful semantics notions and full abstraction
%results. In addition, we apply all this to recover particular (but frequent in practice)
%cases where reasoning in extensional terms is still valid.
Finally Sect. \ref{sec:conclusions} summarizes some conclusions and future work.

\section{Higher-Order Functional-Logic Programs}\label{sec:crwl}
\subsection{Expressions, patterns and programs}
We consider \emph{function} symbols $f,g, \ldots \in FS$, \emph{constructor} symbols
$c,d,\ldots \in CS$, and \emph{variables} $X,Y,\ldots \in {\cal V}$; each $h \in FS \cup CS$ has an associated \emph{arity}, $ar(h) \in \nat$; ${FS}^n$ (resp. ${CS}^n$) is the set of function (resp. constructor) symbols with arity $n$.
The notation $\overline{o}$ stands for tuples
of any kind of syntactic objects $o$.
The set of \emph{applicative expressions} is defined by
$Exp \ni e ::= X~|~h~|~(e_1\ e_2)~$.
As usual, application is left associative and outer parentheses  can be omitted,
so that $e_1\ e_2\ \ldots e_n$ stands for $((\ldots(e_1\ e_2)\ldots)\ e_n)$.
The set of variables occurring in $e$ is written by $var(e)$.
A distinguished set of expressions is that of \emph{patterns} $t,s \in Pat$, defined by:
$t ::  = X~|~c~t_1 \ldots t_n~|~f~t_1 \ldots t_m $,
where $0 \leq n \leq ar(c), 0 \leq m < ar(f)$.
Patterns are irreducible expressions playing the role of \emph{values}.
\emph{FO-patterns}, defined by $\mathit{FOPat} \ni t ::= X~|~c\ t_1\ldots t_n$ $(n=ar(c))$,
correspond to FO constructor terms, representing ordinary non-functional data-values.
%The rest of patterns are HO-patterns and would have, in a typed setting, a type containing
%$\rightarrow$. In particular,
Partial applications of symbols $h\in FS \cup CS$ to other patterns are HO-patterns and can be seen
as truly data-values representing functions from an \emph{intensional} point of view.
Examples of patterns with the signature of Ex. \ref{fdouble} are: \emph{0, s X, s, f', fadd f' f'}.
The last three are HO-patterns. Notice that \emph{f, fadd f f} are not patterns since $f$ is not a
pattern ($ar(f)=0$).

% Expressions $X~e_1 \ldots e_m$ ($m \geq 0$) are called  \emph{flexible}
% (\emph{variable application} when $m > 0$).
% \emph{Rigid} expressions  have the form $h~e_1 \ldots e_m$; moreover, they are \emph{junk} if $h\in CS^n$ and $m>n$,
%  \emph{active} if $h\in FS^n$ and $m\geq n$, and \emph{passive} otherwise.
%
\emph{Contexts}  are expressions with a hole defined as $Cntxt \ni
\con ::= [\ ]~|~\con\ e~|~e\ \con$.  Application of $\con$ to $e$
(written $\con[e]$) is defined by $[\ ][e] = e\ ;\ (\con\ e')[e] =
\con[e]\ e'\ ;\ (e'\ \con)[e] = e'\ \con[e]$. Substitutions $\theta
\in Subst$ are finite mappings from variables to expressions;
$[X_i/e_i, \ldots, X_n/e_n]$ is the substitution which assigns $e_i
\in Exp$ to the corresponding $X_i \in {\cal V}$. We will mostly use
\emph{pattern-substitutions} (or simply \emph{p-substitutions})
$PSubst = \{ \theta\in Subst~|~\theta(X) \in Pat, \forall X\in{\cal V}\}$. %We write $\epsilon$
% for the identity substitution, $dom(\theta)$ for the domain of $\theta$, and
% $vRan(\theta) = \bigcup_{X\in dom(\theta)}var(X\theta)$.

As usual while describing semantics of non-strict languages, we enlarge the signature
with a new 0-ary constructor symbol $\perp$, which can be used to build the sets
$Expr_\perp,Pat_\perp,PSubst_\perp$ of \emph{partial} expressions, patterns and p-substitutions resp.
%Partial expressions are ordered by the {\em approximation} ordering $\sqsubseteq$ defined as the least
%partial ordering satisfying $\perp \sqsubseteq e$ and $e \sqsubseteq e'
%\Rightarrow {\cal C}[e] \sqsubseteq {\cal C}[e']$
%for all $e,e' \in Exp_\perp, {\cal C} \in {\it Cntxt}$. %FUERA This partial ordering can be extended to substitutions: given $\theta,\sigma\in Subst_\bot$ we say $\theta\sqsubseteq\sigma$ if $X\theta\sqsubseteq X\sigma$ for all $X\in\var$.

% A \emph{HOCRWL}-program (or simply a \emph{program}) consists of one or more \emph{program rules} for each $f \in FS^n$,
% having the form
% $f~t_1 \ldots t_n \tor r$ where $(t_1, \ldots, t_n)$ is a linear (i.e. variables occur only once) tuple of (maybe HO) patterns
% and $r$ is any expression.
A \emph{HOCRWL}-program (or simply a \emph{program}) consists of one
or more \emph{program rules} of  the form $f~t_1 \ldots t_n \tor r$
where $f \in FS^n$, $(t_1, \ldots, t_n)$ is a linear (i.e. variables
occur only once) tuple of (maybe HO) patterns and $r$ is any
expression.
Notice that confluence or termination is not required.
% , and that $r$ may have variables not occurring in $f~t_1 \ldots t_n$.
%FUERA (we write $vExtra(R)$ for such variables in a rule $R$).
In the present work we restrict ourselves to programs not containing {\it extra
  variables}, i.e., programs for which $var(r) \subseteq var(f~\overline{t})$ holds for any program rule.
There are technical reasons for such limitation (see Sect.
\ref{sec:discussion}), whose practical impact is on the other hand
mitigated by known extra-variables elimination techniques
\cite{DiosLopezProle06,AntoyHanus06}.
%\todo{The possible
%extension of this work to rules with extra variables will be
%discussed in Sect. \ref{sec:discussion}.}
%
\emph{HOCRWL}-programs often allow also {\em conditions} in the
program rules. However, % under {\it HOCRWL} %the semantics presented in next section
 programs with conditions can be transformed
into equivalent programs without conditions; therefore we consider
only unconditional rules.
% The original {\it HOCRWL}\/ logic  considered also
%  \emph{joinability}  conditions in rules to achieve a better treatment of
% strict equality as built-in, which is a subject orthogonal to the aims of this paper.
% Therefore, we consider only unconditional rules.

Some FLP systems, like Curry, do not  allow HO-patterns in left-hand sides of function definitions. We call \emph{left-FO} programs to these special kind of \emph{HOCRWL}-programs. We remark that all the notions and results in the paper are applicable to \emph{left-FO} programs and we stress the fact that Ex. \ref{fdouble} is one of them.

\subsection{The \emph{HOCRWL} proof calculus \cite{GHR97}}

The semantics of a program ${\cal P}$ is determined in {\it HOCRWL} by means of a proof calculus able to derive
reduction statements of the form $e \rightarrowtriangle t$, with $e \in Exp_\perp$ and $t \in Pat_\perp$,
meaning informally that $t$ is (or approximates to) a possible value of $e$, obtained by
evaluation of $e$ using ${\cal P}$ under call-time choice. %FUERA Besides this logical semantics, {\it HOCRWL} programs come in \cite{GHR97} with a model-theoretic semantics based on applicative algebras, with existence of a least Herbrand model. We will not use this aspect of the semantics here.

The {\it HOCRWL}-proof calculus is presented in Fig. \ref{fig:hocrwl}.
We write $\clho e \crwlto t$ to express that $e\crwlto t$ is derivable in that calculus using
the program ${\cal P}$. The \emph{HOCRWL-denotation} of an expression $e \in Exp_\perp$ is
defined as $\den{e}^{\mathcal{P}}_{\it HOCRWL}=\{t\in Pat_\perp \mid \clho e\crwlto t\}$.
$\cal P$ and {\it HOCRWL} are frequently omitted in those notations.

\begin{figure}[h!]
  \centering
  {\small
    \begin{tabular}{|l|}
      \hline
      {\bf (B)}$\ $
      $\begin{array}{c}
        \\
        \hline
        e\crwlto\bot
      \end{array}$
      \qquad\qquad
      {\bf (RR)}$\ $
      $\begin{array}{c}
        \\
        \hline
        x\crwlto x
      \end{array}$ $\qquad  x\in \var$\\[0.5cm]

      {\bf (DC)}$\ $
      $\begin{array}{c}
        e_1\crwlto t_1\ \ldots\ e_n\crwlto t_m\\
        \hline
        h~e_1 \ldots e_m \crwlto h~t_1 \ldots t_m
      \end{array}$ $\qquad h\in \Sigma, \mbox{ if } h~t_1 \ldots t_m \mbox{ is a partial pattern, } m \geq 0$\\[0.5cm]

      {\bf (OR)}$\ $
      $\begin{array}{c}
%         e_1\crwlto t_1 \ldots\ e_n\crwlto t_n~~r~a_1 \ldots a_m \crwlto t\\
       e_1\crwlto p_1\theta \ldots\ e_n\crwlto p_n\theta~~~r\theta~a_1 \ldots a_m \crwlto t\\
        \hline
        f~e_1 \ldots e_n~a_1 \ldots a_m \crwlto t
      \end{array}$ \qquad
      $ \begin{array}{l}
\mbox{ if } m \geq 0, \theta \in PSubst_\perp \\
(f~p_1\ldots p_n \tor r) \in  \prog
\end{array}$\\[0.5cm]
%         (f~t_1\ldots t_n \tor r) \in [\mathcal{P}]_{\perp}$\\[0.5cm]
      \hline
    \end{tabular}
    }
    \caption{({\it HOCRWL}-calculus)}%\footnote{\emph{G}oal-\emph{O}riented \emph{R}ewriting \emph{C}alculus})}
    \label{fig:hocrwl}
\end{figure}

Looking at in Ex. \ref{fdouble} we have $\den{fdouble\ f\ 0} = \{0,s\ (s\ 0),\perp,s\ \perp,s\ (s\ \perp)\}$
and $\den{fdouble\ f'\ 0} = \{0,s\ 0,s\ (s\ 0),\perp,s\ \perp,s\ (s\ \perp)\}$.

We will use the following  result stating an important compositionality property
of the semantics of \emph{HOCRWL}-expressions: the semantics of a whole expression depends only on the
semantics of its constituents, in a particular form reflecting the idea of call-time choice.
%FUERA The second part of the theorem is a technical result, needed in some proofs, concerning the size of the  involved derivations.

\begin{theorem}[Compositionality of \emph{HOCRWL} semantics, \cite{LRSflops08}]\label{lemacompos}\label{LCompAp}
For any $e \in Exp_\perp$, $\con \in Cntxt$, $\den{\cnn{e}} = \bigcup_{t\in\den{e}}\den{\cnn{t}}$.
% \begin{itemize}
% \item[(i)] $\den{\cnn{e}} = \bigcup_{t\in\den{e}}\den{\cnn{t}}$,
% for any program $\cal P$ and expression $e\in Exp_\perp$.
% \\In other terms, $\cnn{e} \clto t \Leftrightarrow \exists s. (e \clto s \wedge \cnn{s} \clto t)$.
% \item[(ii)] In the ($\Rightarrow$) part of (i), if $t\neq \perp, \con \neq [\ ]$ and the derivation of $\cnn{e} \clto t$ has size
% $K$, then the derivations of $e \clto s$ and $\cnn{s} \clto t$ can be chosen with sizes $< K$
% and $\leq K$ respectively.
% \end{itemize}
\end{theorem}

The {\it HOCRWL} logic is related to several operational notions. In \cite{GHR97} a goal solving narrowing calculus was presented and its strong adecuacy to {\it HOCRWL} shown. The operational semantics of \cite{AHHOV05} has been also used in many works in the field of FLP. Its equivalence with the first order version of {\it HOCRWL} was stated in \cite{LRSentcs06}, and it can be transfered to higher order through the results of \cite{LRSflops08,AHHOV05}. The formalization of graph rewriting of \cite{EchahedJanodet97IMAG,EchahedJanodet98JICSLP} has been often used in FLP too, and although never formally proved, it is usually considered that it specifies the same behaviour. Finally, in \cite{LRSflops08} a notion of higher order rewriting with local bindings called \emph{HOlet-rewriting} and its lifting to narrowing were proposed, and its adequacy to {\it HOCRWL} was formally proved. It can be summarized in the following result:
\begin{theorem}[\cite{LRSflops08}]\label{letrw}
$\forall e \in Exp$, $t \in Pat$, $t \in \den{e}^\prog$ iff $\prog \vdash e\ \f^* t$, where $\f^*$ stands for the reflexive-transitive closure of the {\it HOlet}-rewriting relation.
\end{theorem}
% With this result on hand,
Therefore, %anytime
we can use the set of total values computed for an expression in {\it HOCRWL} as a characterization of the operational behaviour of that expression, as %by this result
it has a strong correspondence, not only with its behaviour under
{\it HOlet}-rewriting, but also under any of the operational notions
metioned above.

\subsection{Extensionality}
In order to achieve more generality and technical precision wrt. the discussion of Ex.\ref{fdouble}, we
% We
introduce here some new terminologies and notations
about extensional equivalence and related notions that will be used later on.
They can be expressed in terms of the {\it HOCRWL} semantics $\den{\_}$.

\begin{definition}[Extensional equivalence, extensional semantics]\label{def:extensional}
\begin{itemize}
\item[(i)] Given $n \geq 0$, two expressions $e,e' \in Expr_\perp$ are said to be
\emph{$n$-extensionally equivalent} ($e \sim_n e'$) iff
$\den{e~e_1\ldots e_n} = \den{e'~e_1\ldots e_n}$, for any $e_1,\ldots,e_n \in Expr_\perp$.
\item[(ii)] Given $n \geq 0$, $e \in Expr_\perp$, the $n$-extensional semantics of $e$ is defined as:
 $\den{e}_{ext_n} = \lambda t_1 \dots \lambda t_n.~\den{e\ t_1 \dots t_n }$ ($t_i \in Pat_\perp$).
\end{itemize}
\end{definition}

We can establish some relationships between these notions:

\begin{proposition}\label{propExt}~%\mbox{~}\\
\begin{itemize}
\item[(i)] $e \sim_n e'$ $\Rightarrow$ $e \sim_m e'$, for all $m > n$.
\item[(ii)] $e \sim_n e'$ $\Leftrightarrow$ $\den{e~t_1\ldots t_n} = \den{e'~t_1\ldots t_n}$,
for any $t_1,\ldots,t_n \in Pat_\perp$.
\item[(iii)] $e \sim_n e'$ $\Leftrightarrow$ $\den{e}_{ext_n}=\den{e'}_{ext_n}$
\end{itemize}
\end{proposition}

\begin{proof}
The proof is easy, %but
thanks to compositionality of $\den{\_}$ (Th. \ref{lemacompos}).
\begin{itemize}
\item[(i)] Assume $e \sim_n e'$, $m > n$, let $e_1 \ldots e_m \in Expr_\perp$.
We must prove $\den{e~e_1\ldots e_m} = \den{e'~e_1\ldots e_m}$. We reason as follows:
$$
\begin{array}{lll}
\den{e~e_1\ldots e_m} &=& \\
\den{(e~e_1\ldots e_n) e_{n+1}\ldots e_m} &=& \mbox{(by compositionality)}\\
\bigcup_{t\in \den{e~e_1\ldots e_n}}{\den{t~e_{n+1}\ldots e_m}} &=& \mbox{(since $e \sim_n e'$)} \\
\bigcup_{t\in \den{e'~e_1\ldots e_n}}{\den{t~e_{n+1}\ldots e_m}} &=& \mbox{(by compositionality)} \\
\den{(e'~e_1\ldots e_n)e_{n+1}\ldots e_m} &=& \\
\den{e'~e_1\ldots e_m}
\end{array}
$$
\item[(ii)] Another direct use of compositionality
\item[(iii)] Consequence of \emph{(i),(ii)} and definitions of $\sim_n,\den{\_}_{ext_n}$.
\end{itemize}
\end{proof}
%%%%%%%%%%%%%%%%%%%%%%%%%%%%%%%%%%%%%%%%%%%%%%%%5

\section{CRWL and Full Abstraction}\label{sec:fully}
%%%%%%%%%%%%%%%%%%%%%%%%%%%%%%%%%%%%%%%%%%%%%%%%5
% \input{fully}

\subsection{Full Abstraction}

In this section we examine technically soundness and full abstraction of the {\it HOCRWL} semantics
$\den{\_}$  and its extensional variants $\den{\_}_{ext_k}$. We can anticipate a positive answer for $\den{\_}$
and negative for the others.

Full abstraction depends on a criterion of observability for expressions.
In constructor based languages, like FLP languages, it is reasonable to observe the outcomes of computations,
given by constructor forms reached by reduction. Here, we can interpret 'constructor form' in a liberal sense,
including HO-patterns, or in a more restricted sense, only with FO-patterns. This leads to the following notions of observation.

\begin{definition}[observations]
Let ${\cal  P}$ be a program. We consider the following observations:
\begin{itemize}
\item  $\obs^{\cal P}: Expr \mapsto Pat$ is defined as
%        $\obs^{\cal P}(e) = \{t\in Pat \mid e \f^*_P t\}$
        $\obs^{\cal P}(e) = \{t\in Pat \mid \prog \vdash e\ \f^* t\}$
\item  %$\obsfo^{\cal P}: Expr \mapsto FOPat$ is defined as
       %$\obsfo^{\cal P}(e) = \{t\in FOPat \mid e \f^*_P t\} (= \obs^{\cal P}(e) \cap FOPat)$
      $\obsfop{\cal P}: Expr \mapsto \mathit{FOPat}$ is defined as
      $\obsfop{\cal P}(e) = \{t\in \mathit{FOPat} \mid \prog \vdash e\ \f^* t\} (= \obs^{\cal P}(e) \cap \mathit{FOPat})$
\end{itemize}
\end{definition}

We remark that, due to the strong correspondence between reduction and semantics given
by  Th. \ref{letrw}, we also have $\obs^{\cal P}(e) = \den{e}^{\cal P} \cap Pat$, implying in particular
$\obs^{\cal P}(e) \subseteq \den{e}^{\cal P}$ (and similar conditions hold for $\obsfo$).
%But notice that these conditions by no means
%guarantee soundness or full abstraction for the semantics.

Now we turn to the definition of full abstraction.
In programming languages like PCF the condition for full abstraction is usually stated as:
\begin{equation}\label{fa0}
   \den{e} = \den{e'}\Leftrightarrow {\cal O}({\cal C}[e]) = {\cal O}({\cal C}[e']), \mbox{for any context ${\cal C}$}
\end{equation}
where ${\cal O}$ is the observation function of interest. Programs do not need to be mentioned, because programs and expressions can be identified by contemplating the
evaluation of $e$ under $\ppp$ as the evaluation of a big $\lambda$-expression or big $let$-expression embodying $\ppp$ and $e$.
Contexts pose no problems either. In our case, since programs  are kept different from
expressions, some care must be taken. % to achieve the desired level of precision.
%In particular, it is not a good idea
%to consider a fixed program $\ppp$ to which refer all elements (semantics, contexts, etc.) involved in \emph{(\ref{fa0})}.
It might happen that $\ppp$ has not enough syntactical elements and rules to built interesting distinguishing contexts.
For instance, if in Ex. \ref{fdouble} we drop the definition of $fdouble$, and we consider $\obsfo$ as observation,
then we cannot built a context that distinguishes $f$ from $f'$.
This would imply that soundness or full abstraction would not be intrinsic to the semantics, but would greatly depend on the program.
What we need is requiring the right part of (\ref{fa0}) to hold for all contexts that might be obtained
by extending $\ppp$ with new auxiliary functions. To be more precise, we say that $\ppp'$ is a \emph{safe extension} of $(\ppp,e)$
if $\ppp' = \ppp\cup \ppp''$, where $\ppp''$ does not include defining rules for any function symbol occurring in $\ppp$ or $e$.
% Any sensible notion of semantics should verify %It is clear that %a safe extension does not change the semantics of the involved expression, that is:
% $\den{e}^{\ppp}=\den{e}^{\ppp'}$ when $\ppp'$ safely extends $(P,e)$.
% %This  happens indeed for all the semantics considered below.
The following property of {\it HOCRWL} regarding safe extensions will be crucial for full abstraction.
The property is  subtler than it appears to be, as witnessed by the fact that it fails to hold if programs have extra variables, as discussed in Sect. \ref{sec:discussion}.
\begin{lemma}\label{lSafeExt}
% For any program $\prog$, expression $e$ and safe extension $\prog'$ for $(\prog, e)$ we have $\den{e}^\prog = \den{e}^{\prog'}$.
$\den{e}^{\ppp}=\den{e}^{\ppp'}$ when $\ppp'$ safely extends $(P,e)$.
\end{lemma}
\begin{proof}
As $\prog \subseteq \prog'$ then $\den{e}^{\ppp} \subseteq \den{e}^{\ppp'}$ trivially holds, as every {\it HOCRWL}-proof for $\prog \vdash e \clto t$ is also a proof for $\prog' \vdash e \clto t$. \\

On the other hand, to prove the inclusion $\den{e}^{\ppp'} \subseteq
\den{e}^{\ppp}$ let us precisely formalize the notion of safe
extension. For any program $\prog$, we write $defs(\prog)$  for the
set of function symbols defined in $\prog$ (i.e., appearing at the
root of some left-hand side of a program rule of $\prog$); for any
expression $e$, we write  $\fs^{e}$ for the set of function symbols
appearing in $e$; for any program $\prog$ and rule $(l \tor r) \in
\prog$ we define $\fs^{(l \tor r)} = \fs^l \cup \fs^r$ and $fs^\prog
= \bigcup_{(l \tor r) \in \prog} \fs^{(l \tor r)}$. Then $\prog'$ is
a safe extension of $(\prog, e)$ iff $\prog' = \prog \uplus \prog''$
such that $defs(\prog'') \cap (\fs^e \cup \fs^\prog) = \emptyset$.

Now we will see that for any proof for $\prog' \vdash a \clto s$ if $defs(\prog'') \cap \fs^a = \emptyset$ then $defs(\prog'') \cap \fs^s = \emptyset$ and for any premise $a' \clto s'$ appearing in that proof we have $defs(\prog'') \cap (\fs^{a'} \cup \fs^{s'}) = \emptyset$, by induction on the structure of $\prog' \vdash a \clto s$. Let us do a case distinction over the rule applied at the root. If it was B then the only statement is $a \clto \perp$ for which the condition holds because $\perp \not\in \fs$. If it was RR then the only statement is $x \clto x$, but $x \not\in \fs$. If it was DC then we apply the IH over each $e_i \clto t_i$, because  $defs(\prog'') \cap \fs^{(h~e_1 \ldots e_m)} = \emptyset$ implies $defs(\prog'') \cap \fs^{e_i}  = \emptyset$ for each $e_i$. All that is left is checking that $defs(\prog'') \cap \fs^{(h~t_1 \ldots t_m)} = \emptyset$. But $defs(\prog'') \cap \fs^{t_i}  = \emptyset$ for each $t_i$ by IH, and $h \in \fs^{(h~e_1 \ldots e_m)} \cap defs(\prog'') = \emptyset$ by hypothesis, so we are done. Finally, for OR we apply the IH to $e_i \clto p_i\theta$ and its premises, as we did in DC. Besides $f \in \fs^{(f~e_1 \ldots e_n~a_1 \ldots a_m)} \cap defs(\prog'') = \emptyset$ by hypothesis, so $(f~p_1 \ldots p_m \tor r) \in \prog$, hence $defs(\prog'') \cap \fs^{(f~p_1 \ldots p_m \tor r)} = \emptyset$, because $\prog''$ is a safe extension. Combining both facts with the absence of extra variables in program rules we get $\fs^{r\theta} \cap defs(\prog'') = \emptyset$. But $\fs^{(f~e_1 \ldots e_n~a_1 \ldots a_m)} \cap defs(\prog'') = \emptyset$ by hypothesis, hence $\fs^{(r\theta~a_1 \ldots a_m)} \cap defs(\prog'') = \emptyset$, to which we can apply the IH to conclude the proof.

Finally, assuming a proof $\prog' \vdash e \clto t$ we may apply the property above because $defs(\prog'') \cap \fs^e = \emptyset$, as $\prog''$ is a safe extension. Therefore $\prog''$ was not used in that proof and so it is also a proof for $\prog \vdash e \clto t$, since $\prog' = \prog \uplus \prog''$.
% ~\\
%
% Now assume a proof for $\prog' \vdash e \clto t$, we will see that for any statement $a \clto s$ appearing in that proof we have $defs(\prog'') \cap (\fs^a \cup \fs^s) = \emptyset$, by induction on the structure of $\prog' \vdash e \clto t$. Let us do a case distinction over the rule applied at the root. If it was B then the only statement is $e \clto \perp$ for which the condition holds as $defs(\prog'') \cap \fs^e = \emptyset$ because $\prog''$ is a safe extension and $\perp \not\in \fs$. If it was RR then the only statement is $x \clto x$. but $x \not\in \fs$. If it was RR then we apply the IH over each $e_i \clto t_i$ and every premise of them, so all we have to check is that $defs(\prog'') \cap \fs^{(h~t_1 \ldots t_m)} = \emptyset$. But $defs(\prog'') \cap \fs^{t_i}  = \emptyset$ for each $t_i$ by IH, and $h \in \fs^{(h~e_1 \ldots e_m)} \cap defs(\prog'') = \emptyset$ because $\prog''$ is a safe extension, so we are done. Finally, for OR we may apply the IH to $e_i \clto p_i\theta$ and its premises, as we did in DC. Besides
% \todo{On the other hand, assume $\prog' \vdash e \clto t$, then by the conditions over safe extensions}
\end{proof}

We can now define:

\begin{definition}[Full abstraction]\label{def:fullabst}\mbox{~}\\

\noindent{(a)}~
A semantics is \emph{fully abstract} wrt ${\cal O}$ iff for any ${\ppp}$ and $e,e' \in Expr$, the following two conditions
are equivalent:
\\
 \indent  {(i)}~ $\den{e}^{\ppp} = \den{e'}^{\ppp}$
  ~~~
  {(ii)}
  \begin{tabular}[t]{l}
  ${\cal O}^{\ppp'}({\cal C}[e]) = {\cal O}^{\ppp'}({\cal C}[e'])$ for any
  $\ppp'$ safely  extending \\$(\ppp,e)$, $(\ppp,e')$
  and any ${\cal C}$ built with the signature of $\ppp'$.
  \end{tabular}
\\
%In words: semantic equality is equivalent to observational indistinguishability.

\noindent{(b)}~
A  notion weaker than full abstraction is: a semantics is \emph{sound} wrt ${\cal O}$ iff the condition \emph{(i)} above implies
the condition \emph{(ii)}.
\\
%In words: semantic equality implies observational indistinguishability.

%\noindent{(c)}~
%A semantics is \emph{ compositional} iff for any ${\ppp}$ and $e,e' \in Expr$, the following two conditions
%are equivalent:
%\\~~
%\indent  {(i)}~ $\den{e}^{\ppp}\! =\! \den{e'}^{\ppp}$
%  ~
%  {(ii)}~ $\den{{\cal C}[e]}^{\ppp}\! =\! \den{{\cal C}[e']}^{\ppp}$ for any  ${\cal C}$ built with the signature of $\ppp$.
%\\
%In words: the semantics of an expression depends only on the semantics of its subexpressions. Notice that $(ii) \Rightarrow (i)$
%holds trivially (take ${\cal C} = [~]$).
%We remark also that $\den{\_}$ is compositional iff
%is fully abstract wrt itself when taken as observation function.
% OJO: borrado porque ya no cuadra con la nueva definici\'on
\end{definition}

For extensional semantics, our Ex. \ref{fdouble} (and obvious generalizations to
arities $k>1$) constitutes a proof of the following negative result:

\begin{proposition}\label{unsound}
\begin{itemize}
For any $k>0$, $\den{\_}_{ext_k}$is unsound wrt $\obs,\obsfo$.
This remains true even if programs are restricted to be left-FO.
\end{itemize}
\end{proposition}

This contrast with the following:

\begin{theorem}[Full abstraction]\label{th:fullabstraction}%\mbox{~}\\
% \begin{itemize}
$\den{\_}$ is fully abstract wrt $\obs$ and $\obsfo$.
% \end{itemize}
\end{theorem}

The proof for this theorem will be based on the compositionality of $\den{\_}$ and the following result:
\begin{lemma}\label{lAuxFADen}
Let $\prog$ be any program. Consider the transformation $\hat{\_} :
Pat_\perp \rightarrow Pat$  defined by:
$$
%\begin{array}{ll}
\hat{X} = X ~~~~ ~~~~ \hat{\perp} = bot ~~~~~~~
\widehat{h~t_1~\ldots t_m} = h~\hat{t_1} ~\ldots~\hat{t_m}
%\end{array}
$$
where $bot$ is a fresh constant constructor symbol.
%Note that $\hat{t} \in Pat$ for any $t \in Pat_\perp$.
Consider also the program $\prog' = \prog \uplus \prog_{g_t}$,
where $\prog_{g_t}$ consists of the
following rules defining some fresh symbols $g_s \in FS$:
$$
\begin{array}{ll}
g_X~U \tor U ~~~~ ~~~~ g_\perp~X \tor bot \\
g_{(h~t_1~\ldots t_m)} (h~X_1~\ldots X_m) \tor h~(g_{t_1} X_1) \ldots (g_{t_m} X_m) \\
\end{array}
$$
Then:
\begin{itemize}
\item[(i)] $\prog'$ is a safe extension of $(\prog, e)$.
\item[(ii)] $t\in \den{e}^\prog$ iff $\hat{t} \in \den{g_{t}~e}^{\prog'}$,
for any $e \in Exp_\perp, t \in Pat_\perp$ built with the signature of $\prog$.
\end{itemize}
%For any program $\prog$, the transformation $\hat{\_} : Pat_\perp \rightarrow Pat$ is defined by:
%$$
%\begin{array}{ll}
%\hat{X} = X ~~~~ ~~~~ \hat{\perp} = bot \\
%\widehat{h~t_1~\ldots t_m} = h~\hat{t_1} ~\ldots~\hat{t_m}
%\end{array}
%$$
%where $bot$ is a fresh (not used in $\prog$ nor any expression so far) constant constructor symbol. Note that then $\forall t \in Pat_\perp$ we have $\hat{t} \in Pat$. Besides for any $e \in Exp_\perp, t \in Pat_\perp$ built with the signature of $\prog$, $t \in \den{e}^\prog$ iff $\hat{t} \in \den{g_{t}~e}^{\prog'}$, where $\prog' = \prog \uplus \prog_{g_t}$ for $\prog_{g_t}$ containing the following rules defining some fresh symbols $g_s \in FS$:
%$$
%\begin{array}{ll}
%g_X~U \tor U ~~~~ ~~~~ g_\perp~X \tor bot \\
%g_{(h~t_1~\ldots t_m)} (h~X_1~\ldots X_m) \tor h~(g_{t_1} X_1) \ldots (g_{t_m} X_m) \\
%\end{array}
%$$
%Besides $\prog'$ is a safe extension of $(\prog, e)$.
\end{lemma}
\begin{proof}
It is clear that $\prog'$ is a safe extension as it only defines new rules for fresh function symbols. The other equivalence holds by two simple inductions on the structure of $t$.
%\todo{Demo si tengo sitio?}
% Now, concerning the left to right implication, assuming $t \in \den{e}^\prog$, we will see that $\hat{t} \in \den{g_{t}~e}^{\prog'}$ by induction on the structure of $t$. Concerning the base cases:
% \begin{itemize}
%  \item $t \equiv \perp$ : Then
% $$
% \infer[OR]{\prog' \vdash g_\perp~e \clto bot \equiv \hat{\perp}}
%           { \infer[B]{e \clto \perp}{} \ \infer[DC]{bot \clto bot}{} }
% $$
% \item $t \equiv X \in \var$ : Then
% $$
% \infer[OR]{\prog' \vdash g_X~e \clto X \equiv \hat{X}}
%           {\infer{e \clto X}{(*)} \ \infer[RR]{X \clto X}{}  }
% $$
% $(*)$ as $t \equiv X \in \den{e}^\prog = \den{e}^{\prog'}$ by hypothesis, and because $\prog'$ is a safe extension.
% \end{itemize}
\end{proof}

\begin{proof}[For Theorem \ref{th:fullabstraction}]
First of all we will prove the full abstraction wrt. $\obs$.
We will see that $\den{e}^\prog = \den{e'}^\prog$ iff for any safe extension $\prog'$ of $(\prog,e)$ and $(\prog,e')$, for any context $\con$ built with the signature of $\prog'$ we have $\obs^{\prog'}(\con[e]) = \obs^{\prog'}(\con[e'])$.\\
Concerning the left to right implication, assume $\den{e}^\prog = \den{e'}^\prog$ and fix some safe extension $\prog'$ and some context $\con$ built on it. First we will see that $\obs^{\prog'}(\con[e]) \subseteq \obs^{\prog'}(\con[e'])$. Assume some $t \in \obs^{\prog'}(\con[e])$, then $t \in \den{\con[e]}^{\prog'}$ by definition and Th. \ref{letrw}. But then
$$
\begin{array}{ll}
t \in \den{\con[e]}^{\prog'}
= \bigcup_{t \in \den{e}^{\prog'}} \den{\con[t]}^{\prog'} & \mbox{ by Th. \ref{lemacompos}} \\
= \bigcup_{t \in \den{e}^{\prog}} \den{\con[t]}^{\prog'} & \mbox{ by Lemma \ref{lSafeExt}, as $\prog'$ is a safe extension} \\
= \bigcup_{t \in \den{e'}^{\prog}} \den{\con[t]}^{\prog'}  & \mbox{ by hypothesis} \\
= \bigcup_{t \in \den{e'}^{\prog'}} \den{\con[t]}^{\prog'}  & \mbox{ by Lemma \ref{lSafeExt}, as $\prog'$ is a safe extension } \\
= \den{\con[e']}^{\prog'}  & \mbox{ by Th. \ref{lemacompos}} \\
\end{array}
$$
But then $t \in \obs^{\prog'}(\con[e'])$ by definition and Th. \ref{letrw}. The other inclusion can be proved in a similar way.

% Regarding the right to left implication, we will use the transformation $\hat{\_} : Pat_\perp \rightarrow Pat$ defined by:
% $$
% \begin{array}{ll}
% \hat{X} = X ~~~~ ~~~~ \hat{\perp} = bot \\
% \widehat{h~t_1~\ldots t_m} = h~\hat{t_1} ~\ldots~\hat{t_m}
% \end{array}
% $$
% where $bot$ is a fresh constant constructor symbol. Note that then $\forall t \in Pat_\perp$ we have $\hat{t} \in Pat$. Besides we can prove that for any program $\prog$, $e \in Exp_\perp, t \in Pat_\perp$ built with the signature of $\prog$, $t \in \den{e}^\prog$ iff $\hat{t} \in \den{g_{t}~e}^{\prog'}$, where $\prog' = \prog \uplus \prog_{g_t}$ for $\prog_{g_t}$ containing the following rules defining some fresh symbols $g_t \in FS$:
% $$
% \begin{array}{ll}
% g_X~U \tor U ~~~~ ~~~~ g_\perp~X \tor bot \\
% g_{(h~t_1~\ldots t_m)} (h~X_1~\ldots X_m) \tor h~(g_{t_1} X_1) \ldots (g_{t_m} X_m) \\
% \end{array}
% $$
% This can prove this property by a simple induction on the structure of $t$. Note also that then $\prog'$ is a safe extension and then we can assume $\obs^{\prog'}(\con[e]) = \obs^{\prog'}(\con[e'])$ for any $\con$ built on $\prog'$. Thus, for any $t \in \den{e}^\prog$ we have $\hat{t} \in \den{g_t~e}^{\prog'}$, and so $\hat{t} = \obs^{\prog'}(g_t~e) = \obs^{\prog'}(g_t~e')$ by definition, Th. \ref{letrw} and hypothesis. But then $\hat{t} \in \den{g_t~e'}^{\prog'}$ and so $t \in \den{e'}^{\prog}$
Regarding the right to left implication, we will use the transformation $\hat{\_}$ of Lemma \ref{lAuxFADen}. We can also take the program $\prog'$ of Lemma \ref{lAuxFADen} which is a safe extension of $(\prog,e)$ and $(\prog,e')$ as it only defines new rules for fresh function symbols. Therefore we can assume $\obs^{\prog'}(\con[e]) = \obs^{\prog'}(\con[e'])$ for any $\con$ built on $\prog'$. Besides, for any $t \in \den{e}^\prog$ we have $\hat{t} \in \den{g_t~e}^{\prog'}$ by Lemma \ref{lAuxFADen}, and so $\hat{t} \in \obs^{\prog'}(g_t~e) = \obs^{\prog'}(g_t~e')$ by definition, Th. \ref{letrw} and hypothesis. But then $\hat{t} \in \den{g_t~e'}^{\prog'}$ by definition and Th. \ref{letrw}, and so $t \in \den{e'}^{\prog}$ by Lemma \ref{lAuxFADen} again. The other inclusion of $\den{e'}$ in $\den{e}$ can be proved in a similar way.\\

Now we will prove the full abstraction wrt. $\obsfo$. The left to right implication can be proved in exactly the same way we did for $\obs$. Concerning the other implication we modify the transformation $\hat{\_}$ of Lemma \ref{lAuxFADen} in the following way:
$$
\begin{array}{ll}
\widehat{h~t_1~\ldots t_m} = h_m~\hat{t_1} ~\ldots~\hat{t_m} \\
g_{(h~t_1~\ldots t_m)} (h~X_1~\ldots X_m) \tor h_m~(g_{t_1} X_1) \ldots (g_{t_m} X_m) \\
\end{array}
$$
where $h_m$ is a fresh constructor symbol of arity $m$. Note that then $\forall t \in Pat_\perp$ we have $\hat{t} \in \mathit{FOPat}$. Besides it is still easy to prove that for any $e \in Exp_\perp, t \in Pat_\perp$ built with the signature of $\prog$, $t \in \den{e}^\prog$ iff $\hat{t} \in \den{g_{t}~e}^{\prog'}$, where $\prog' = \prog \uplus \prog_{g_t}$, and that $\prog'$ is a safe extension of $\prog$, by a trivial modification of the proof for Lemma \ref{lAuxFADen}. With these tool the proof proceeds exactly like in the one for $\obs$, but using these new definitions of $\hat{\_}$ and $g_t$.
\end{proof}

\subsection{Discussion: the case of extra variables}\label{sec:discussion}
%%%%%%%%%%%%%%%%%%%%%%%%%%%%%%%%%%%%%%%%%%%%%%%%5
% \input{discussion}
% Contraejemplos varios:\\
As pointed in Sect. \ref{sec:crwl}, in this work we assume that our programs do not contain extra variables, i.e., $var(r) \subseteq var(f~\overline{t})$ holds for any program rule $f~t_1 \ldots t_n \tor r$. This condition is necessary for the full abstraction results to hold, as we can see in the following example.
\begin{example}
Consider a signature such that  $FS = \{f/1, g/1\}$, $CS = \{0/0, 1/0\}$, and the program
 $\prog = \{f~X \tor Y~X\}$.
Note the extra variable $Y$ in the rule for $f$.

Then we have $\den{f~0}^\prog =
\{\perp\} = \den{f~1}^\prog$, because any derivation of $f~0 \clto t$
using (OR)  must have the form
$$
\infer[OR]{\prog \vdash f~0 \clto t}
{
0 \clto 0
\ ~~\infer[X]{\varphi~0 \clto t}
            {\ldots }
}
$$
where $\varphi$ can be any pattern ($f$, $g$, $0$, $1$ or $\perp$) %$f,g$
and X can be (OR) or (B). In all cases the only possible value for $t$ in $\varphi~0 \clto t$ will be $\perp$.
A similar reasoning holds for $f~1$.
However, for
% Ademas como $f~0$ y $f~1$ son ground sus hipersemanticas tb son iguales por el mismo motivo.
$\prog' = \prog \uplus \{g~0 \tor 1\}$, which is a safe extension for $(\prog, f~0)$ and $(\prog, f~1)$ we can do:
$$
\infer[OR]{~~~~~~\prog' \vdash f~0 \clto 1~~~~~~}
{
0 \clto 0
\ ~~\infer[OR]{g~0 \clto 1}
            {0 \clto 0 \ ~~1 \clto 1 }
}
$$
while for $f~1$ we can only do:
$$
\infer[OR]{\prog' \vdash f~1 \clto \perp}
{
1 \clto 1
\ ~~\infer[B]{g~1 \clto \perp}
            {}
}
$$
Hence the context $[]$ and the safe extension $\prog'$ yield different observations for $f~0$ and $f~1$.
\end{example}
The previous example  can be discarded if we assume that we have at least one constructor for each arity, or at least for the maximum of the arities of function symbols. This is reasonable because it is like having tuples of any arity. With this assumption and the previous program and expression we do not have $\den{f~a}^\prog = \den{f~b}^\prog$ anymore, as $c~a \in \den{f~a}$ and $c~b \in \den{f~b}$, hence the hypothesis of the condition for full abstraction fails.\\
Nevertheless the following example shows that full abstraction fails even under the assumption of having a constructor for each arity.
% Sin embargo el siguiente contraejemplo funciona aun con tuplas
\begin{example}
For $\prog = \{f~1 \tor 2, h~X \tor f~(Y~X)\}$ and $FS = \{f/1, h/1, g/1\}$ we have $\forall \theta \in PSubst_\perp, 1 \not\in \den{(\theta(Y))~0}^\prog \cup \den{(\theta(Y))~1}^\prog$, hence $\den{h~0}^\prog = \{\perp\} = \den{h~1}^\prog$. But for $\prog' = \prog \uplus \{g~0 \tor 1\}$, which is a safe extension for $(\prog, h~0)$ and  $(\prog, h~1)$, we have $\prog' \vdash h~0 \clto 2$ while $\prog' \vdash h~1 \not\clto 2$.
% luego ni siquiera con la signatura fijada tenemos que la hipersemantica sea full abstract respecto a si misma.
\end{example}

% The point is that, if extra variables are allowed, for a fixed program $\prog$ and a pair of expressions $e, e'$ we cannot ensure that for any safe extension $\prog'$ for $(\prog, e)$ and $(\prog, e')$ it holds that
The point is that, if extra variables are allowed, for a fixed program $\prog$ and an expression $e$ we cannot ensure that for any safe extension $\prog'$ for $(\prog, e)$ it holds that $\den{e}^\prog = \den{e}^{\prog'}$; i.e., Lemma \ref{lSafeExt} does
not hold. We cannot even grant that $\den{e}^\prog = \den{e'}^\prog$ implies that $\den{e}^{\prog'} = \den{e'}^{\prog'}$ for any safe extension $\prog'$, which in fact is what it is needed for full abstraction, and what we have exploited in the (counter-)examples above.
It is also relevant that both examples are left-FO programs, and therefore the problems do not come
from the presence of higher order patterns in function definitions. %It seems that the intensional view of functions of {\it HOCRWL} has some inherent limitations that should
%Therefore, the extension of this work to cover extra variables, and to give a more active role to variables in general, is a challenging subject of future work.

As a conclusion of this discussion, we contemplate the extension of this work to cope with extra variables as a challenging subject of future work.

\section{Conclusions and Future Work}\label{sec:conclusions}
%%%%%%%%%%%%%%%%%%%%%%%%%%%%%%%%%%%%%%%%%%%%%%%%5
% \input{conclusions}
%\todo{No basta con borrar cosas, hay q adaptarlo todo, pero los dos parrafos del final ya sobran}
% Let us start with some facts, followed by some opinions.
We have seen that reasoning extensionally in existing FLP languages with HO nondeterministic functions is
not valid in general (Ex. \ref{fdouble}, Prop. \ref{unsound}). In contrast, thinking in intensional
functions is not an arbitrary exoticism, but rather an appropriate point of view for that setting
(Th. \ref{th:fullabstraction}). We stress the fact that adopting an intensional view
of the \emph{ meaning } of functions is compatible with a discipline of programming in which
programs are restricted %obliged
 to be left-FO, that is, the use of HO-patterns in left-hand sides of program rules
is forbidden. This is the preferred choice by  some people in the FLP community, mostly because
HO-patterns in left-hand sides cause
some problems to the type system. Our personal opinion is the following: since HO-patterns appear in
the semantics even if they are precluded from programs, they could be freely permitted, at least as far as they
are compatible with the type discipline. There are quite precise works \cite{GHR99} pointing out which are
the problematic aspects, mainly \emph{opacity} of patterns. Existing systems could incorporate
restrictions, so that only type-safe uses of HO-patterns are allowed. More work could be done
along this line.

We have seen in Sect. \ref{sec:discussion} how the presence of extra variables in programs
destroys full-abstraction of the {\it HOCRWL} semantics. Recovering it for such family of programs
is an obvious subject of future work.
Another very interesting, and  somehow related matter,   is giving variables a more active role
in the semantics.
Certainly, the results in the paper are not restricted to ground expressions,
but their interest  for expressions having
variables is limited by the fact that in the notions of semantics and observations  considered in the paper,
variables are implicitly treated as generic constants.
For instance, the expressions $e_1 \equiv %double X$
X + X$ and $e_2 \equiv X+0$ do have the same semantics $\den{\_}_\perp$
($\den{e_1}_\perp  = \den{e_2}_\perp = \{\perp\}$).
%Something similar happens with the other notions of semantics given before.
Full abstraction of $\den{\_}_\perp$ ensure that
$\obs({\cal C}[e_1]) = \obs({\cal C}[e_2])$ for any context $\cal C$.
This is ok as far as one is only interested in possible reductions starting from $e_1,e_2$.
If this is the case, certainly $e_1$ and $e_2$ have equivalent behavior (no successful reduction to
a pattern can be done with any of them). However, in some sense $e_1$ and $e_2$ have different `meanings',
that are reflected in different behaviors; for instance, if $e_1$ and $e_2$ are subject to narrowing,
or if $e_1$ and $e_2$ are used as right hand sides in a program rule.

\medskip

\noindent{\bf Acknowledgments}~~We are grateful to Rafa Caballero for his intense
collaboration while developing this research.

\bibliographystyle{abbrv}

% \bibliography{../biblio}
\end{document}